\documentclass[runningheads,a4paper]{llncs}

\usepackage{amsfonts}
\usepackage{amsmath}
\usepackage{amssymb}
\usepackage[T1]{fontenc}
\usepackage{algorithm}
\usepackage[noend]{algpseudocode}
\usepackage{graphicx}
\usepackage{tabularx}
\usepackage[caption=false]{subfig}

\newcommand{\Geqt}{\ensuremath{G_\text{eqt}}}

\begin{document}

\title{Leader Election and Shape Formation with Self-Organizing Programmable Matter}

\titlerunning{Leader Election and Shape Formation}

%
%
\author{
  Joshua J. Daymude\inst{1}
  \and Zahra Derakhshandeh\inst{1} \and
  Robert Gmyr\inst{2}
  Thim Strothmann\inst{2} \and
  Rida Bazzi\inst{1} \and
  Andr\'ea W.\ Richa\inst{1} \and
  Christian Scheideler\inst{2}
}

\authorrunning{Daymude et al.}

\institute{Department of Computer Science and Engineering,\\
    Arizona State University, USA, \\
   \email{\{jdaymude,zderakhs,bazzi,aricha\}@asu.edu }
   \and
   Department of Computer Science,\\
    University of Paderborn, Germany, \\
   \email{\{gmyr,thim\}@mail.upb.de, scheideler@upb.de}}
   

%
%

\maketitle

\begin{abstract}
  We consider programmable matter consisting of
  simple computational elements, called {\em particles}, that can establish and release
  bonds and can actively move in a self-organized way, and we investigate
  the feasibility of solving fundamental problems relevant for programmable matter.
  As a suitable model for such self-organizing particle systems,
  we will use a generalization of the geometric amoebot model first proposed
  in SPAA 2014.
  Based on the geometric model, we present efficient local-control algorithms for leader
  election and line formation requiring only particles with constant size
  memory, and we also discuss the limitations of solving these problems within
  the general amoebot model.
\end{abstract}

\section{Introduction}

A central problem for programmable matter is shape formation, and various
solutions have already been discovered for that problem using different approaches
like DNA tiles \cite{patitz2014}, moteins \cite{CDBG11}, or nubots
\cite{winfree13}. We are studying shape formation using the amoebot model
that was first proposed in~\cite{spaa-ba14}. In order to determine how
decentralized shape formation can be handled, we are particularly interested
in the connection between leader election and shape formation. In the leader
election problem we are given a set of particles, and the problem is to select
one of the particles as the leader. Many problems like the consensus problem (all
particles have to agree on some output value) can easily be solved once the
leader election problem can be solved. The same has also been observed for
shape formation, as most shape formation algorithms depend on some seed
element. However, the question is whether shape formation can even be solved
in circumstances where leader election is not possible. The aim of this paper
is to shed some light on the dependency between leader election and shape
formation: On one hand, we present an efficient decentralized algorithm for solving leader election in a geometric variant of the amoebot model and show how having a leader leads to an efficient solution for the basic line formation formation problem; on the other hand, we show that both problems cannot be solved efficiently under the general version of our model .
Before we present our results in more detail, we first give a formal definition of the model
and the problems we study in this paper.

\subsection{Models} 
\label{sec:model}

We use two models throughout this work. Firstly, we consider a generalization
of the amoebot model~\cite{spaa-ba14} which abstracts from any geometry
information. We call this model the \emph{general amoebot model}. Secondly, we
consider a model that is essentially equivalent to the original amoebot model presented in~\cite{spaa-ba14}
but is defined based on the general amoebot model. We refer to this second
model as the \emph{geometric amoebot model}.

In the \emph{general amoebot model}, programmable matter consists of a uniform
set of simple computational units called particles that can move and bond to
other particles and use these bonds to exchange information. The particles act
asynchronously and they achieve locomotion by expanding and contracting, which
resembles the behavior of amoeba.

As a base of this model, we assume that we have a set of particles that aim at
maintaining a connected structure at all times. This is needed to prevent the
particles from drifting apart in an uncontrolled manner like in fluids and
because in our case particles communicate only via bonds. The shape and
positions of the bonds of the particles mandate that they can only assume
discrete positions in the particle structure. This justifies the use of a
possibly infinite, undirected graph $G=(V,E)$, where $V$ represents all
possible positions of a particle (relative to the other particles in their
structure) and $E$ represents all possible transitions between positions.

Each particle occupies either a single node or a pair of adjacent nodes in
$G$, i.e., it can be in two different {\em shapes}, and every node can be
occupied by at most one particle. Two particles occupying adjacent nodes are
\emph{connected}, and we refer to such particles as \emph{neighbors}.
Particles are \emph{anonymous} but the bonds of each particle have unique
labels, which implies that a particle can uniquely identify each of its
outgoing edges. Each particle has a local memory, and any pair of connected
particles has a shared memory that can be read and written by both particles.

Particles move by \emph{expansions} and \emph{contractions}: If a
particle occupies one node (i.e., it is \emph{contracted}), it can expand to
an unoccupied adjacent node to occupy two nodes. If a particle occupies two
nodes (i.e., it is \emph{expanded}), it can contract to one of these nodes to
occupy only a single node. Performing movements via expansions and
contractions has various advantages. For example, it would easily allow a
particle to abort a movement if its movement is in conflict with other
movements. A particle always knows whether it is contracted or expanded and this
information will be available to neighboring particles. In a \emph{handover},
two scenarios are possible: a) a contracted particle $p$ can "push" a
neighboring expanded particle $q$ and expand into the neighboring node
previously occupied by $q$, forcing $q$ to contract, or b) an expanded
particle $p$ can "pull" a neighboring contracted particle $q$ to a cell
occupied by it thereby expanding that particle to that cell, which allows $p$
to contract to its other cell. The ability to use a handover allows the system
to stay connected while particles move (e.g., for particles moving in a
worm-like fashion). Note that while expansions and contractions may represent
the way particles physically move in space,
they can also be interpreted  as a particle "looking ahead" and
establishing new logical connections (by expanding) before it fully moves to a new
position and severs the old connections it had (by contracting).

Summing up over all assumptions above, the {\em state} of a particle is
uniquely determined by its shape, the contents of its local memory, the edges
it has to neighboring particles, the contents of their shared memory (which
may allow a particle to obtain further information about the neighboring
particles beyond their shape), and finally the shape of the neighboring
particles. The {\em state of the particle system} (or short, {\em system
state}) is defined as the combination of all particle states.
We say a particle system in a system state in which the particle occupy a set of nodes $A \subseteq V$ is \emph{connected} if the graph $G |_A$ induced by $A$ is connected.
We assume the standard asynchronous computation model, i.e., only one particle can be active
at a time. Whenever a particle is active, it can perform an {\em action}
(governed by some fixed, finite size program controlling it) consisting of a
finite amount of computation (involving its local memory, the
shared memories with its neighboring particles, and random bits) followed by no or a single
movement. Hence, a {\em computation} of a particle system is a potentially
infinite sequence of actions $A_1, A_2, \ldots$ based on some initial system
state $s_0$, where action $A_i$ transforms system state $s_{i-1}$ into system
state $s_i$.
The (parallel) time complexity of a computation is usually
measured in {\em rounds}, where a round is over once every particle has
been given the chance to perform at least one action.

Let ${\cal S}$ be the set of all system states in which the particle system is
connected. In general, a {\em computational problem} $P$ for the particle
system is specified by a set ${\cal S}' \subseteq {\cal S}$ of permitted
initial system states and a mapping $F:{\cal S}' \rightarrow 2^{\cal S}$,
where $F(s) \subseteq {\cal S}$ determines the set of permitted {\em final}
states for any initial state $s \in {\cal S}'$. A particle system {\em solves}
problem $P=({\cal S}',F)$ if for any initial system state $s \in {\cal S'}$,
all computations of the particle system eventually reach a system state in
$F(s)$ without losing connectivity, and whenever such a system state is
reached for the first time, the system stays in $F(s)$.
If for all computation a final state is reached in which all
particles decided to halt (i.e., they decided not to perform any further
actions, irrespective of future events), then the particle system is also said
to {\em decide} problem $P$. Note that being in a final state does not
necessarily mean that all particles decided to halt.
If ${\cal S}'={\cal S}$, so {\em any} initial state is permitted (including
arbitrary faulty states, as long as the particle system is connected), then a
particle system solving $P$ is also said to be {\em self-stabilizing}. It is
well-known that in general a distributed system solving a problem $P$ cannot
decide it and also be self-stabilizing because if so, it would often be
possible to come up with an initial state $s$ where a member of the system
decides to halt prematurely, disallowing the system to eventually reach a
state in $F(s)$.

Besides the general amoebot model, we will also consider the \emph{geometric
amoebot model}. The geometric amoebot model is a specific variant of the
general amoebot model in which the underlying graph $G$ is defined to be the
equilateral triangular graph $\Geqt$ (see Figure~\ref{fig:graphAndBoundaries}),
and the bonds of the particles are labeled in a consecutive way in clockwise orientation
around a particle so that every particle has the same sense of clockwise
orientation. However, we do not assume that the labeling is uniform, so the particles do not necessarily share a common sense of direction in the grid.
\begin{figure}[htb]
  \includegraphics{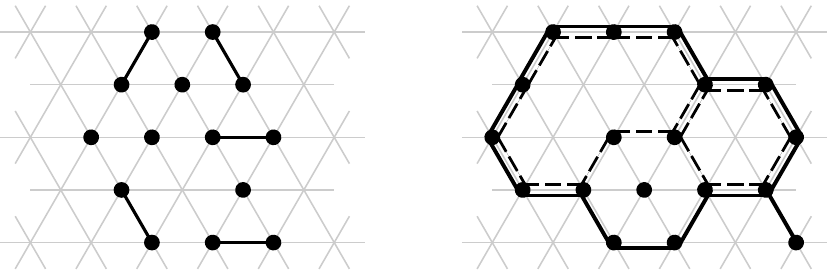}
  \caption{
    The left part shows an example of a particle structure in the geometric amoebot model.
    A contracted particle is depicted as a black dot,
    and an expanded particle is depicted as two black dots connected by an edge.
    The right part shows a particle structure with 3 boundaries.
    The outer boundary is shown as a solid line and the two inner boundaries are shown as dashed lines.
  }
  \label{fig:graphAndBoundaries}
\end{figure}

\subsection{Problems}

In this paper we consider the following two problems.
For both problems we define the set of initial system states as the set of all states such that the particle system is connected and all memories are empty.

For the \emph{leader election problem} the set of final system states contains
any state in which the particles form a connected structure and exactly one
particle is a leader (i.e., only this particle is in a leader state while the
remaining particles are in a non-leader state). Our goal will be to come up
with a distributed algorithm that allows a particle system to decide the
leader election problem. Note that the leader election problem is well defined
for both the general amoebot model and the geometric amoebot model.

In a \emph{shape formation problem}, the set of final states consists of those
system states where the particle structure forms the desired shape. As a
specific example of a shape formation problem, we consider the \emph{line
formation problem}. In the geometric amoebot model, the shape the particles
have to form is a straight line in the equilateral triangular grid and all particles have to be contracted in a final system state.
Of course, in the general amoebot model a straight line is not well-defined. Hence, for
this model the set of final states for the line formation problem is defined
to consist of all system states in which the particles form a simple path in
$G$.

Throughout the paper, we assume for the sake of simplicity that in an initial state all particles are contracted.
Our algorithms can easily be extended to dispose of this assumption.

\subsection{Our Contributions}

We focus on the problem of solving leader election and shape
formation for particles with {\em constant size memory}. For shape formation, we
just focus on the already mentioned line formation problem.

For the {\em geometric} amoebot model,
we show that there is a {\em distributed algorithm that can decide the leader
election problem} (Section~\ref{sec:geometric}), i.e., at the end we have exactly one leader and the leader
knows that it is the only leader left. Moreover, the runtime for our leader
election algorithm is {\em worst-case optimal} in the sense that it needs
{\em $\mathcal{O}(L_\text{max})$ rounds with high probability} (w.h.p.)\footnote{By {\em with high probability}, we mean with probability at least $1-1/n^c$, where $n$ is the
number of particles in the system and $c>0$ is a constant.}, where $L_\text{max}$ is the maximum length of a boundary
between the particle structure and an empty region (inside or outside of it)
in $\Geqt$.
Based on the leader election algorithm, we present a {\em distributed algorithm that solves the line formation problem} with \emph{worst-case optimal work} in Section~\ref{sec:line}.

On the other hand, for the {\em general} amoebot model, we show that {\em neither leader election nor
shape formation can be decided by any distributed algorithm} in Section~\ref{sec:impossibility}. More concretely, we show that there cannot exist a randomized algorithm for solving either problem with any bounded probability of success in the general model.

The algorithms presented for leader election and line formation under the geometric amoebot model both assume that
the system is in a well-initialized state.
It would certainly be desirable to
have algorithms that can tolerate any initial state, but at the end of the
paper, in Section~\ref{sec:ss}, we show that there are certain limitations to solving leader election
and line formation in a self-stabilizing fashion.

\subsection{Related Work}

Many approaches related to programmable matter have recently been proposed.
One can distinguish between active and passive systems. In passive systems the
particles either do not have any intelligence at all (but just move and bond
based on their structural properties or due to chemical interactions with the
environment), or they have limited computational capabilities but cannot
control their movements. Examples of research on \emph{passive systems} are
DNA computing \cite{Adl94,BDLS96,CDBG11,DPSS11,WLWS98}, tile self-assembly
systems in general  (e.g., see the surveys in~\cite{doty2012,patitz2014,Woods2013intrinsic}),
population protocols \cite{AAD+06}, and slime molds~\cite{BMV12,LTT+10}.
We will not describe these models in detail as they are only of little relevance for our approach.
On the other hand in \emph{active systems}, computational
particles can control the way they act and move in order to solve a specific task.
Robotic swarms and modular robotic systems are
some examples of active programmable matter systems.

In the area of \textit{swarm robotics} it is usually assumed that there is a
collection of autonomous robots that have limited sensing, often including
vision, and communication ranges, and that can freely move in a given area.
They follow a broad variety of goals: for example, graph exploration
(e.g.,~\cite{fl13}), gathering problems (e.g., \cite{AG3,ci12}), shape
formation problems (e.g.,~\cite{fl08,kilobots}), and to understand the global
effects of local behavior in natural swarms like social insects, birds, or
fish (e.g.,~\cite{DBLP:1211-1909,Cha09}). Surveys of recent results in
swarm robotics can be found in~\cite{Ker12,McL08}; other samples of
representative work can be found in~\cite{AR10,BFMS11,CP08,DFSY10,DS08,HABFM02,KM11}. While the
analytical techniques developed in the area of swarm robotics and natural
swarms are of some relevance for this work, the individual units in those
systems have more powerful communication and processing capabilities than in
the systems we consider.

The field of \textit{modular self-reconfigurable robotic systems} focuses on intra-robotic aspects
such as the design, fabrication, motion planning, and control of autonomous kinematic machines with variable morphology
(see e.g.,~\cite{FNKB88,YSS+07}).
\textit{Metamorphic robots}  form a subclass of self-reconfigurable robots that share some of the characteristics of our geometric model~\cite{Chi94}.
The hardware development in the field of self-reconfigurable robotics has been complemented
by a number of algorithmic advances (e.g.,~\cite{BKRT04,kilobots,WWA04}),
but so far
mechanisms that automatically scale from a few to hundreds or thousands of individual units are still under investigation,
and no rigorous theoretical foundation is available yet.

The \emph{nubot} model~\cite{chen2014fast,chen2013parallel,winfree13} by Woods et al.
aims at providing the theoretical
framework that would allow for
a more rigorous algorithmic study of biomolecular-inspired systems,
more specifically of self-assembly systems with active molecular
components. Although bio\-molecular-inspired systems share  similarities with our self-organizing particle systems, there are many differences that do not allow us to translate the algorithms and other results under the nubot model to our systems --- e.g.,
there is always an arbitrarily large supply of "extra" particles that can be added to the system as needed, and the system allows for an additional (non-local) notion of rigid-body movement.

Our developed leader election algorithm for the geometric amoebot model shares some similarities with the algorithm of~\cite{BeckersW01} for cellular automata. However, since cellular automata work in a synchronous fashion and have access to a global compass, the approach of~\cite{BeckersW01} is vastly different from our leader election algorithm.

\section{Impossibility Results in the General Amoebot Model}
\label{sec:impossibility}
In this section, we show that both leader election and line formation are impossible to solve in the general amoebot model.
Suppose that
there is a distributed algorithm solving the line formation problem in the
general amoebot model (when starting in a well-initialized state).
Since in
this case it is possible to decide when $G|_{A'}$ forms a line, it is also
possible to design a protocol that solves the leader election problem:
once the line has been formed, its two endpoints contend for
leadership using tokens with random bits sent back and forth until one of them
wins. On the other hand, one can deduce from~\cite{DBLP:conf/focs/ItaiR81}
that in the general amoebot model there is no distributed algorithm that can
decide when a leader has been elected (with any reasonable success
probability).

More concretely, in~\cite{DBLP:conf/focs/ItaiR81} the authors
show that for a ring of anonymous nodes there is no algorithm that can
correctly decide the leader election problem\footnote{Or, in their words, that can
solve the leader election problem with distributive termination.} with any
probability $\alpha>0$, i.e., for any algorithm in which the particles are
guaranteed to halt, the error probability is unbounded. Since in the general
amoebot model $G$ can be any graph, we can set $G$ to be a ring whose size is
the number of particles and the result of~\cite{DBLP:conf/focs/ItaiR81} is
directly applicable.

Hence, there is no a distributed algorithm deciding
the line formation problem (with any reasonable success probability) in the
general amoebot model, and therefore not even an algorithm for solving it since a protocol solving the problem could easily be transformed into a protocol deciding it.

\section{Leader Election in the Geometric Amoebot Model}
\label{sec:geometric}

In this section we show how the leader election problem can be decided in the geometric amoebot model.
Our approach organizes the particle system into a set of cycles and executes
an algorithm on each cycle independently (Section~\ref{sec:organizationIntoCycles}).
For simplicity and ease of presentation we first state the protocol in a simple model in which particles have a global view of the cycle they are part of, act synchronously, and have unbounded local memory (Section~\ref{sec:algorithmLE}).
We prove its correctness in Section~\ref{sec.algoCorrectness}.
In Section~\ref{sec.algoLocalRealization}, we present the corresponding local-control protocol that works without these assumptions in the geometric amoebot model.
However, since the local realization combines many different token passing schemes and techniques, a formal analysis would be beyond the scope of this paper.
Instead, we show that our approach has expected linear runtime in a variant of the simple model which is motivated by the token passing schemes of the local-control protocol, i.e., it takes into account that interaction between two particles is dependent on the distance between those particles (Section~\ref{sec.algoRuntime}).
We conclude our analysis of the leader election problem by presenting an extension of our protocol (in the simpler model) in Section~\ref{sec:whp} with a linear runtime with high probability.

\subsection{Organization into Cycles}
\label{sec:organizationIntoCycles}
Let $A \subseteq V$ be any initial distribution of contracted particles such that $\Geqt|_A$ is connected.
Consider the graph $\Geqt|_{V \setminus A}$ induced by the unoccupied nodes in $\Geqt$.
We call a connected component of $\Geqt|_{V \setminus A}$ an \emph{empty region}.
Let $N(R)$ be the neighborhood of an empty region $R$ in $\Geqt$.
Then all nodes in $N(R)$ are occupied and we call the graph $\Geqt|_{N(R)}$ a \emph{boundary}.
Since $\Geqt|_A$ is a connected finite graph, exactly one empty region has infinite size while the remaining empty regions have finite size.
We define the boundary corresponding to the infinite empty region to be the unique \emph{outer boundary} and refer to a boundary that corresponds to a finite empty region as an \emph{inner boundary}, see Figure~\ref{fig:graphAndBoundaries}.

The particles occupying a boundary can instantly (i.e., without communication) organize themselves into a cycle using only local information:
Consider a boundary corresponding to an empty region $R$.
Let $p$ be a particle occupying a node $v$ of the boundary.
By definition there exists a non-occupied node $w \in R$ that is a adjacent to $v$ in the graph $\Geqt$.
The particle $p$ iterates over the neighboring nodes of $v$ in clockwise orientation around $v$ starting at $w$.
Consider the first occupied node it encounters; the particle occupying that node is the successor of $p$ in the cycle corresponding to that boundary.
Analogously, $p$ finds its predecessor in the cycle by traversing the neighborhood of $v$ in counter-clockwise orientation.

Note that a single particle can belong to up to three boundaries at once.
Furthermore, a particle cannot locally decide whether two empty regions it sees (i.e., maximal connected components of non-occupied nodes in the neighborhood of $v$) are distinct.
We circumvent these problems by letting a particle treat each empty region in its local view as distinct.
For each such empty region, a particle executes an independent instance of the same algorithm.
Hence, we say a particle acts as a number of (at most three)  distinct \emph{agents}.
For each of its agents a particle determines the predecessor and successor as described above.
This effectively connects the set of all agents into disjoint cycles as depicted in Figure~\ref{fig:agents}.
Observe that from a global perspective the cycle of the outer boundary is oriented clockwise while a cycle of an inner boundary is oriented counter-clockwise.
This is a direct consequence of the way the predecessors and successors of an agent are defined.

\begin{figure}[htb]
  \includegraphics{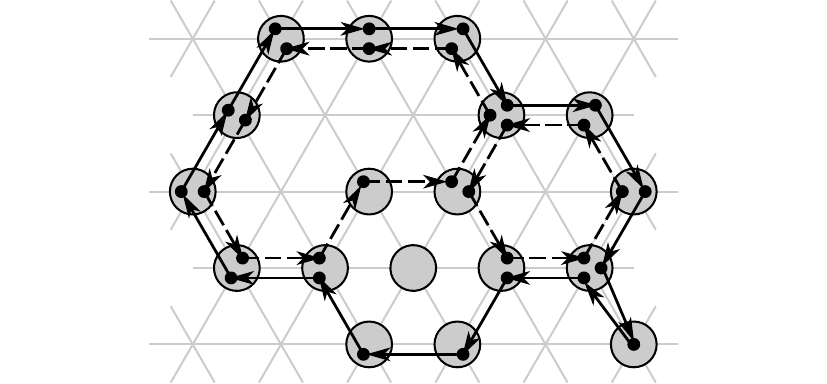}
  \caption{
    The depicted particle system is the same as in the right part of Figure~\ref{fig:graphAndBoundaries}.
    In this figure particles are depicted as gray circles.
    The black dots inside of a particle represent its agents.
    As in Figure~\ref{fig:graphAndBoundaries} the outer boundary is solid and the two inner boundaries are dashed.
  }
  \label{fig:agents}
\end{figure}

\subsection{Algorithm}
\label{sec:algorithmLE}
The leader election algorithm operates independently on each cycle.
At any given time, some subset of agents on a cycle will consider themselves \emph{candidates}, i.e. potential leaders of the system.
Initially, every agent considers itself a candidate.
Between any two candidates on a cycle there is a (possibly empty) sequence of non-candidate agents.
We call such a sequence a \emph{segment}. For a candidate $c$ we refer to the segment coming after   $c$ in the direction of the cycle as $seg(c)$ and refer to its length by $|seg(c)|$.
We refer to the candidate coming after $c$ as the \emph{succeeding candidate} ($succ(c)$) and to the candidate coming before $c$  as the \emph{preceding candidate} ($pred(c)$) (see Figure~\ref{fig:segment}).
We drop the $c$ in parentheses if it is clear from the context.
We define the distance $d(c_1, c_2)$ between candidates $c_1$ and $c_2$ as the number of agents between $c_1$ and $c_2$ when going from $c_1$ to $c_2$ in \emph{direction} of the cycle.
We say a candidate $c_1$ \emph{covers} a candidate $c_2$ (or $c_2$ \emph{is covered by} $c_1$) if $|seg(c_1)| > d(c_2, c_1)$ (see Figure~\ref{fig:segment}).
\begin{figure}[htb]
  \includegraphics{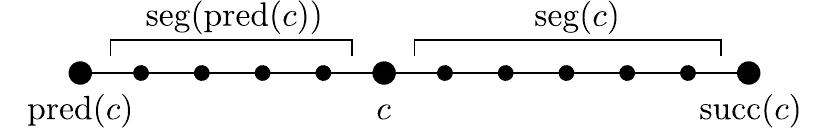}
  \caption{
  	The figure depicts a part of a cycle that is oriented to the right.
    Non-candidate agents are small black dots, candidates are bigger dots.
    The candidate $c$ covers $pred(c)$ since $|seg(c)| > d(pred(c), c)$.
  }
  \label{fig:segment}
\end{figure}
The leader election progresses in \emph{phases}.
In each phase, each candidate executes Algorithm~\ref{alg:geometricLE}.
A phase consists of three synchronized \emph{subphases}, i.e., agents can only progress to the next subphase once all agents have finished the current subphase.

\begin{algorithm}[htb]
  \caption{Leader Election for a Candidate $c$}
  \label{alg:geometricLE}
  \begin{algorithmic}
    \Statex \textbf{Subphase 1:}
    \State $pos$ $\leftarrow$ position of $succ(c)$
    \If{covered by any candidate or $|seg(c)| < |seg(pred(c))|$ }
      \State \Return{not leader}
    \EndIf
    \\
    \Statex \textbf{Subphase 2:}
    \If {coin flip results in heads}
    \State transfer candidacy to agent at $pos$
    \EndIf
    \\
    \Statex \textbf{Subphase 3:}
    \If{only candidate on boundary}
      \If{outside boundary}
        \State \Return{leader}
      \Else
        \State \Return{not leader}
      \EndIf
    \EndIf
  \end{algorithmic}
\end{algorithm}

Consider the execution of Algorithm~\ref{alg:geometricLE} by a candidate $c$.
If the algorithm returns "not leader" then $c$ revokes its candidacy and becomes part of a segment.
If the algorithm returns "leader", $c$ will become the leader of the particle system.
The transferal of candidacy in subphase 2 means that $c$ withdraws its own candidacy but at the same time promotes the agent at position $pos$ (i.e., $succ(c)$ in subphase 1) to be a candidate.
Once a candidate becomes a leader it broadcasts this information such that all particles can halt.

\subsection{Correctness}
\label{sec.algoCorrectness}
In order to show the correctness of our algorithm, we show that it satisfies the following conditions, that relate to the entire particle system (not just a single cycle):
\begin{enumerate}
  \item \emph{Safety}: There always exists at least one candidate.
  \item \emph{Liveness}: In each phase if there is more than one candidate, at least one candidate withdraws leadership with a probability that is bounded below by a positive constant.
\end{enumerate}

\begin{lemma}
  \label{lem:safety}
  Algorithm~\ref{alg:geometricLE} satisfies the safety condition.
\end{lemma}

\begin{proof}
  We will show by induction that on the cycle associated with the outer boundary there will always be at least one candidate.
  Initially, this holds trivially.
  So assume that it holds before a phase.
  Let $c$ be the candidate with the longest segment.
  Then there is no candidate covering $c$ and also $|seg(c)| < |seg(pred(c))|$ cannot be true.
  Hence, $c$ will not withdraw candidacy in subphase 1.
  In subphase 2, the candidacy of $c$ might be transferred but will not vanish.
  Let $c'$ be the agent that received the candidacy if it was transferred and $c' = c$ otherwise.
  In subphase 3, $c'$ will not withdraw candidacy because it lies on the outer boundary.
  Hence, there is still a candidate after the phase. \qed
\end{proof}

\begin{lemma}
  \label{lem:liveness}
  Algorithm~\ref{alg:geometricLE} satisfies the liveness condition.
\end{lemma}

\begin{proof}
  Assume that there are two or more candidates in the system.
  First we consider the case that there is a cycle with two or more candidates.
  If there are segments of different lengths on that cycle, we have $|seg| < |seg(pred)|$ for at least one candidate which will therefore withdraw its candidacy in subphase 1.
  If all segments are of equal length, we have that in subphase 2 with probability at least $\frac{1}{4}$ there is a candidate $c$ that transfers candidacy while $succ(c)$ does not.
  Hence, the number of candidates is reduced with probability at least $\frac{1}{4}$.
  Now consider the case that all cycles have at most one candidate.
  Then there is a cycle corresponding to an inner boundary that has exactly one candidate.
  That candidate will withdraw candidacy in subphase 3 and thereby reduce the number of candidates in the system. \qed
\end{proof}

The following Theorem is a direct consequence of Lemmas~\ref{lem:safety} and~\ref{lem:liveness}.

\begin{theorem}
  Algorithm~\ref{alg:geometricLE} successfully decides the leader election problem.
\end{theorem}

\subsection{Runtime Analysis}
\label{sec.algoRuntime}
For a cycle of agents let $L$ be the length of the cycle and let $l_i$ be the longest segment length before phase $i$ of the execution of Algorithm~\ref{alg:geometricLE}.
We define $l_i = L$ if there is no candidate on the cycle.
It is easy to see that if $l_i \ge L/2$ then in phase $i + 1$ either the leader is elected (outer boundary) or all candidates on the cycle vanish (inner boundary).
For the case $l_i < L/2$, Lemma~\ref{lem:segmentDoubling} provides the key insight of our analysis.

\begin{lemma}
  \label{lem:segmentDoubling}
  For any phase $i$ such that $l_i < L/2$ it holds $l_{i + 1} \ge l_i$ in any case and $l_{i + 1} \ge 2 l_i$ with probability at least $1/4$.
\end{lemma}

\begin{proof}
  Consider a candidate $c$ such that $|seg(c)| = l_i$.
  Subphase 1 can only increase segment lengths and $c$ will not withdraw leadership.
  So after Subphase 1 we have $|seg(c)| \ge l_i$.
  Also, we have $|seg(pred(c))| \ge l_i$ because $c$ covers any candidate $c'$ such that $d(c', c) < l_i$.
  Here, we use $pred(c)$ to refer to the cadidate preceeding $c$ \emph{after} the execution of Subphase 1.
  For Subphase 2 we have to distinguish between two cases based on the outcome of the coin flip of $c$.
  If $c$ does not transfer candidacy, we still have $|seg(c)| \ge l_i$.
  For the case that $c$ does transfer candidacy,
  we have to further distinguish two cases based on the outcome of the coin flip of $pred(c)$.
  If $pred(c)$ also transfers its candidacy, $c$ will receive that candidacy
  while $c$ itself transfers its candidacy forward by a distance of $l_i$.
  Therefore, we still have $|seg(c)| \ge l_i$.
  If $pred(c)$ does not transfer candidacy, after Subphase 2 we have $|seg(pred(c))| \ge 2l_i$ because the segment of $pred(c)$ now spans both the segment of $c$ after Subphase $1$ and the segment of $pred(c)$ after Subphase 1,
  which is at least $l_i$.
  The probability that $c$ transfers candidacy while $pred(c)$ does not is $1/4$.
  \qed
\end{proof}

Let $L_\text{max}$ be the length of the longest cycle in the particle system.
Based on Lemma~\ref{lem:segmentDoubling} it is easy to see that under complete synchronization of subphases and with the agents having a global view of the cycle, our algorithm requires on expectation $\mathcal{O}(\log(L_\text{max}))$ phases to elect a leader.
For the following analysis assume that phases still progress in lockstep for all agents; however, the duration of a phase is dependent on the longest segment length, i.e., phase $i$ requires $\mathcal{O}(l_i)$ rounds.
Theorem~\ref{thm:runtime} gives a bound on the number of rounds required by the algorithm based on this assumption.

\begin{theorem}
  \label{thm:runtime}
  Algorithm~\ref{alg:geometricLE} requires $\mathcal{O}(L_\text{max})$ rounds on expectation.
\end{theorem}

\begin{proof}
  Let the random variable $X_i$ describe the number of rounds during the execution of Algorithm~\ref{alg:geometricLE} such that $l_i \in [2^{i - 1}, 2^i)$.
  Then, under the assumption that phase $i$ requires $\mathcal{O}(l_i)$ rounds, the total runtime of our algorithm is
  \[
    T = \sum_{i = 1}^{\lceil \log(L_\text{max}) \rceil} X_i \cdot O(2^i).
  \]
  Since $E(X_i) \le 4$ due to Lemma~\ref{lem:segmentDoubling}, the expected runtime is
  \[
    E(T) \le \sum_{i = 1}^{\lceil \log(L_\text{max}) \rceil} E(X_i) \cdot O(2^i) = O(L_\text{max}).
  \]
  \qed
\end{proof}
Note that subphase 1 of the algorithm is not important in terms of correctness. However, it is crucial to achieve a linear runtime in expectation. If agents would only execute subphases 2 and 3, the runtime would degrade to $\mathcal{O}(L_\text{max}\log{L_\text{max}})$.

\subsection{Asynchronous Local-Control Protocol}
\label{sec.algoLocalRealization}
In this section we present a realization of Algorithm~\ref{alg:geometricLE} as an asynchronous local-control protocol.
The protocol heavily relies on token passing.
All tokens used by the protocol are messages of constant size and at any time an agent has to hold at most a constant number of tokens.
The tokens of each subphase of Algorithm~\ref{alg:geometricLE} are independent of each other,
so an agent has to handle distinct tokens for each phase at the same time.
If not otherwise specified, we assume that tokens of a single subphase move through the agents in a pipelined fashion
(i.e., a token does not surpass another token in front of it but waits until the agent is free to hold it).

\subsubsection{Candidate Elimination via Segment Comparison}
\label{sec:candidateElimination}
In Subphase 1 of Algorithm \ref{alg:geometricLE}, a candidate $c$ can become demoted if either the length of its front segment is strictly less than that of its back segment, or some other candidate covers $c$ while performing its own comparison. Since $c$ is unable to measure its segments' lengths directly in a local-control setting, it will instead use a token passing scheme to match agents in its front segment with agents in its back segment. Conceptually, this matching proceeds as follows: each agent in the front segment generates a \emph{cover token} which travels into the back segment, matching one by one with the agents therein. If the front segment is longer than the back segment, then some cover token will match with the preceding candidate; a candidate which matches with a cover token is said to be \emph{covered}. If $c$ detects that it was unable to cover any preceding candidates in this process and that the segments lengths were not equal, it concludes that its front segment is too short and revokes candidacy.

In detail, consider $c$ as it begins the segment comparison subphase. Since this subphase requires tokens to be passed in both its front and back segments, we introduce the notion of tokens being either \emph{active} or \emph{passive} to avoid potential collisions with tokens from other candidates. This can be imagined as two different ``lanes'' of tokens which can pass by one another unhindered. Candidate $c$ must first measure the length of its front segment; this is achieved by generating a passive \emph{starting token} and forwarding it along its front segment. Each non-candidate that receives this starting token forwards it and generates a passive cover token which is forwarded back towards $c$. The starting token is ultimately forwarded to the succeeding candidate which consumes it and generates a \emph{final cover token}.

These cover tokens are forwarded back to $c$, which then converts them to active cover tokens when forwarding them to its back segment where they begin to match with the agents therein. Non-candidates in the back segment consume the first cover token that is passed to them and continue to forward the others. If another candidate, say $c'$, receives an active cover token, then $c'$ has been covered by $c$ and revokes candidacy, henceforth behaving as a non-candidate. However, $c'$ may have been executing a number of now irrelevant leader election operations in its front segment; thus, it generates a passive \emph{cleaning token} which is forwarded along its front segment, deleting any tokens it encounters with the exception of those from $c$, which are active. This cleaning token is consumed by $c$. Additionally, it is possible that $c'$ performed work in its back segment, e.g., segment comparison. Therefore, it is possible that there are upcoming non-candidates in its back segment which are now holding incorrect information (e.g. the consumption of a cover token of $c'$). As this incorrect information could interfere with the active cover tokens of $c$, also an active cleaning token is generated to reset all non-candidates in the back segment of $c'$, which is consumed by the first live candidate it encounters. Lastly, this active cleaning token is also responsible for destroying any other cleaning tokens it encounters, as the segment comparison operations of $c$ are the only ones still relevant to this segment.

Eventually, all active cover tokens will be consumed, completing the matching originally described. Candidate $c$ must now gather whether or not it has covered another candidate, which is achieved as follows: instead of matching with the final cover token, the final agent encountered, say $a$, consumes it and generates a \emph{final cleaning token} which travels in the direction of the cycle towards $c$. This final cleaning token records whether or not $a$ was a candidate. At every other agent it visits, it resets the agent's matched state and checks if the agent is a covered candidate. Thus, when $c$ receives this final cleaning token, one of three cases occurs: (1) if the final cleaning token indicates that $c$ covered some other candidate, $c$ remains a candidate; otherwise (2) if the final cleaning token indicates that $a$ was a candidate, $c$ remains a candidate since its segments are of equal length; otherwise (3) $c$'s front segment was too short, and it revokes candidacy.

\subsubsection{Coin Flipping and Candidate Transferral}
\label{sec:canidateTransferral}
In order to realize the second subphase as a local-control protocol, we need a token passing scheme for the candidacy transferral, since a candidate is not able to transfer its candidacy in an instant.
Moreover, since candidates do not have a global view of the system, they cannot know the position of their successor.
The local-control protocol consists of two different token passing schemes and candidates use one of the two schemes dependent on the result of their coin toss.

A candidate $c$ that flips a coin and receives heads sends a \emph{candidacy token} along $seg(c)$.  However, $c$ itself does not give up its candidacy immediately, but (virtually) stays a candidate.
This token is forwarded by non-candidates in $seg(c)$.
A candidate $c'$ that receives a candidacy token sends a \emph{confirmation token} back along the segment of its predecessor. This confirmation token is forwarded back to $c$ such that the virtual copy of $c$ can finally give up its candidacy.
Moreover, there can be three different scenarios for a candidate $c'$ that receives a candidacy token: (i) if $c'$ is not in subphase 2 of the protocol, it continues with its desired behavior, (ii) if it received tails in the coin flip, $c'$ proceeds to the solitude verification (with some caveats that will be explained in the next paragraph), or (iii) if it received heads in the coin flip (i.e. $c'$ is a virtual candidate copy and therefore also aims at transferring leadership), $c'$ will not give up its candidacy (i.e., it will progress to solitude verification once it receives its own confirmation token).
Consequently, it is possible that even though a candidate receives heads, it will not give up candidacy, because its predecessor also received heads.

In addition to the above mentioned token passing scheme for candidates that receive heads, there is also an additional simple scheme for candidates that receive tails. Before progressing to solitude verification a candidate $c$ that receives tails sends a token to $pred(c)$ which is simply sent back by $pred(c)$ and therefore traverses $seg(pred(c))$ twice.

The first token passing scheme makes sure that candidacy tokens are not forwarded infinitely often, but are eventually received by some candidate.
This behavior is due to the fact that candidacy is not revoked immediately after a coin flip, i.e., virtual candidates are able to receive candidacy tokens.

The second scheme guarantees that a candidate $c^*$ with tails synchronizes with its preceding candidate, i.e., $c^*$ waits for the return of its own token and a preceding candidate that flipped heads has the chance to send its own candidacy token to $c^*$ before $c^*$ progresses to solitude verification.
This busy-waiting-like behavior is needed due to the asynchronicity of the amoebot model which could, in a worst-case scenario, prevent progress in the leader election process.
To be more specific consider a scenario with only two remaining candidates $s,t$ with equal segment lengths on a boundary.
Leader election will not make any progress if both candidates continuously get the same result from tossing a coin.
In fact, without the second scheme it is possible to enforce that both candidates always get the same result.
Imagine that candidate $s$ gets heads, while $t$ gets tails in some round $r$.
Since particles act in an asynchronous fashion we can assume that $t$ does not perform any action for any amount of time while $s$ does.
Without the second scheme $s$ can progress to solitude verification (which will fail), then start segment comparison (which will not make any progress) and will flip a coin in again.
We continue this process of going through all the phases until $s$ eventually gets tails.
Note that since $t$ does not perform any action and  due to the definition of an asynchronous round, the round counter does not progress (i.e. the particle system is still in round $r$) and we enforced  that both candidates get the same coin toss result.
With the second scheme, $s$ cannot immediately progress to solitude verification, but waits for its token.
Thereby, $t$ is able to send its candidacy token to $s$ before $s$ can flip a coin for the next time.

\subsubsection{Solitude Verification}
\label{sec:solitudeVerification}
The local-control protocol for solitude verification is based on the following simple observation:
A candidate $c$ is the only candidate left on a cycle if and only if $succ(c) = c$.
To allow the candidates to check this, let each particle assign a unique identifier from $\{1, 2, 3\}$ to each of its agents.
It is easy to see that $succ(c) = c$ if and only if
\begin{enumerate}
  \item $c$ and $succ(c)$ occupy the same node in $\Geqt$, and
  \item $c$ and $succ(c)$ have the same identifier.
\end{enumerate}
Note that since a particle can hold multiple agents, Condition 2 is in fact necessary.
It can be checked using a trivial token passing scheme.
Checking Condition 1 requires a bit more effort:
Intuitively, $c$ enforces its own orientation on all agents in its segment to establish a common coordinate system.
For each agent $a$ in the segment consider the vector pointing from $a$ to $succ(a)$.
The sum of these vectors is $(0, 0)$ if and only if $c$ and $succ(c)$ occupy the same node in $\Geqt$,
see Figure~\ref{fig:solitudesExample}.
\begin{figure}[htb]
  \includegraphics{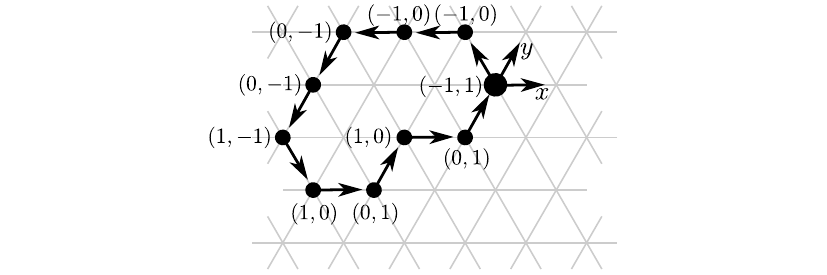}
  \caption{
    An example of solitude verification:
    The candidate (shown slightly bigger) enforces a coordinate system ($x$ and $y$ arrows) on all agents.
    The agents determine the vectors pointing to the succeeding agent in direction of the cycle
    (arrows and tuples at nodes).
    Since the candidate is the only candidate on the cycle, the vectors add up to $(0, 0)$.
  }
  \label{fig:solitudesExample}
\end{figure}

This algorithmic idea can be implemented as a local-control protocol using token passing in the following way.
We use two different types of tokens: \emph{matching tokens} and a unique \emph{activation token}.
The activation token is created by the candidate $c$ and traverses $seg(c)$ four times.
In its first pass, the activation token moves forward (i.e., in direction of the cycle) from $c$ to $succ(c)$
and establishes the common coordinate system among the agents.
Also, whenever the activation token is passed forward in the first pass by an agent $a$,
that agent will create a matching token which stores the vector pointing from $a$ to $succ(a)$ in the common coordinate system.
Once the activation token reaches $succ(c)$,
it initiates its second pass in which it simply moves unhidered backwards (i.e., opposite to the direction of the cycle)
along $seg(c)$ from $succ(c)$ to $c$.
In its third and fourth pass, the activation token again moves forward from $c$ to $succ(c)$ and back.
However, in these two final passes it is not allowed to surpass any matching token.

The matching tokens move from their initial position forwards to $succ(c)$ and then backwards towards $c$.
Every agent is only allowed to store one matching token moving forward and one matching token moving backward
(the direction of movement is stored as part of a token).
Furthermore, a matching token is never allowed to surpass the activation token or another matching token.
Whenever an agent that holds two matching tokens (one for each direction) is activated,
that agent will try to \emph{match} these tokens.
In this matching, the agent compares the vectors stored in the tokens coordinate-wise.
If one vector has a $1$ in a coordinate while the other vector has a $-1$ in that coordinate,
the agent will change both these values to $0$.
Should a token be left with the vector $(0, 0)$ because of this matching process, it is deleted.
Finally, the activated agent passes on any remaining tokens if possible.

It is not hard to see that Condition 1 holds if and only if all matching tokens are deleted in this process.
Furthermore, $c$ can easily distinguish whether all matching tokens have been deleted:
First, note that $c$ will eventually be able to delete or forward the matching token it created itself,
so assume that $c$ does not hold that token anymore.
If some matching token remains, that token will reach $c$ before the activation token returns from its fourth pass
because the activation token cannot surpass the matching token.
On the other hand, if no matching token remains, no such token can reach $c$ before the activation token returns from its fourth pass.
Finally, note that the described process only stores a constant amount of information in a token
and every agent holds at most a constant number of tokens at any time.

\subsubsection{Inner Outer Boundary Test}
The last candidate of a cycle can decide whether its cycle corresponds to an inner or the outer boundary as follows.
A cycle corresponding to an inner boundary has counter-clockwise rotation
while a cycle corresponding to the outer boundary has clockwise rotation, see Figure~\ref{fig:agents}.
The candidate sends a token along the cycle that sums the angles of the turns the cycle takes,
see Figure~\ref{fig:angleSum}.
\begin{figure}[htb]
  \includegraphics{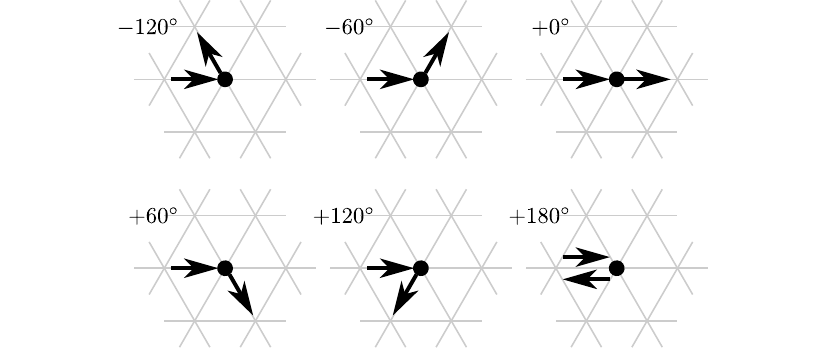}
  \caption{
    The angle between the directions a token enters and exits an agent.
  }
  \label{fig:angleSum}
\end{figure}
When the token returns to the candidate, its value represents the external angle of the polygon
corresponding to the cycle while respecting the rotation of the cycle.
So it is $-360^\circ$ for an inner boundary and $360^\circ$ for the outer boundary.
The token can represent the angle as an integer $k$ such that the angle is $k \cdot 60^\circ$.
Furthermore, to distinguish the two possible final values of $k$ it is sufficient to store the $k$ modulo $5$,
so that the token only requires $3$ bits of memory.

\subsection{Linear Runtime with High Probability}
\label{sec:whp}

The algorithm presented in Section~\ref{sec:algorithmLE} guarantees that a
leader will be elected in an expected linear number of rounds. A small
modification of the algorithm  can lead
to linear runtime {\em with high probability (w.h.p.)}, without
compromising its correctness. For this we only need a slight modification of
Algorithm~\ref{alg:geometricLE}: Subphase 1 needs to be executed twice, and in
Subphase 2 a candidate $c$ generates a sequence of random bits
$b(c)=(b(c)_1,b(c)_2,\ldots)$ that are compared with the random bits of
neighboring candidates instead of just flipping a single coin by themselves.
For two candidates $c$ and $c'$, $b(c)<b(c')$ if there is an $i\ge 0$ so that
$b(c)_j = b(c')_j$ for all $j \le i$ and $b(c)_{i+1} < b(c')_{i+1}$.

\begin{algorithm}[htb]
  \caption{Modified Leader Election}
  \label{alg:geometricLE2}
  \begin{algorithmic}
    \Statex \textbf{Execute Subphase 1 of Algorithm~\ref{alg:geometricLE} twice}
    \\
    \Statex \textbf{Subphase 2:}
    \State generate random bits for $b(c)$ until $b(c)\not=b(pred(c))$
    and $b(c)\not=b(succ(c))$
    \If {$b(c)<b(pred(c))$ or $b(c)<b(succ(c))$}
    \State \Return{not leader}
    \EndIf
    \\
    \Statex \textbf{Execute Subphase 3 of Algorithm~\ref{alg:geometricLE}}
  \end{algorithmic}
\end{algorithm}

In a low-level implementation of Subphase 2, each candidate $c$ continues
to produce random bits for $b(c)$ and sends them in both directions in a
pipelined fashion until it learns that the competition on both sides is over,
i.e., $b(c) \not= b(pred(c))$ and $b(c)\not= b(succ(c))$. The competition is
realized by pairing off the bits from consecutive candidates at the agent
of the segment between them where they meet until a bit pair is different (if both bits are equal, they are deleted at the agent),
which is then reported back to the candidates.

As before, let $l_i$ be the longest segment length before phase $i$, and let
$c$ be a candidate such that $|seg(c)| = l_i$. Subphase 1 can only increase segment
lengths and hence after one execution of Subphase 1, $c$ will not withdraw its leadership candidacy. So after one application of
Subphase 1 we have $|seg(c)| \ge l_i$, and also, $|seg(pred(c))| \ge l_i$.
After another application of subphase 1, various cases can happen:
\begin{itemize}
  \item Case 1: $c$ is not a leader candidate any more. Then the segments of $c$
  and $pred(c)$ now belong to one candidate $c'$, which means that $|seg(c')|
  \ge 2l_i$.

  \item Case 2: $c$ is still a leader candidate, but $pred(c)$ is not a
  candidate any more. Then the segments of $pred(c)$ and $pred(pred(c))$ (by
  which we mean the predecessor of $pred(c)$ after the second application of
  Subphase 1) now belong to $pred(pred(c))$, and since $|seg(pred(pred(c)))| \ge
  |seg(pred(c))|$ before $pred(c)$ gave up its candidacy, after Subphase 1,
  $|seg(pred(pred(c)))| \ge 2l_i$.

  \item Case 3: Both $c$ and $pred(c)$ are still candidates. Note that then, after the second application of Subphase 1,
  also $|seg(pred(pred(c)))| \ge l_i$. Since Subphase 2 ensures that of any two
  consecutive candidates, one of them will give up its candidacy, either $c$ or
  $pred(c)$ will give up its candidacy in Subphase 2, which means that there
  will be a candidate $c'$ with $|seg(c')| \ge 2l_i$.
\end{itemize}
Hence, in each phase the longest segment length is {\em guaranteed} to at
least double, which means that the modified leader election algorithm is
guaranteed to terminate after at most $\log n$ phases. It is easy to show via
Chernoff bounds that in each phase and for each candidate $c$, $\mathcal{O}(\log n)$
bits suffice w.h.p. so that $b(c)\not=b(pred(c))$ and $b(c)\not=b(succ(c))$,
which means that phase $i$ takes at most $\mathcal{O}(l_{i+1} + \log n)$ rounds w.h.p.
Summing up these bounds over all phases results in a runtime of $\mathcal{O}(L_\text{max})$ w.h.p.

Note that while we are confident that the asynchronous local-control algorithm presented in the previous section
performs close to the given simplified analytical bound,
it might require considerable effort to implement the modifications presented in this section as an asynchronous local-control protocol in such a way that the given bound holds.
The main issue is that our analysis requires that the executions of Subphase~1 by Algorithm~\ref{alg:geometricLE2}
are completely synchronized among agents, which appears to be quite tricky to realize.

\section{Line Formation in the Geometric Amoebot Model} \label{sec:lineFormation}
\label{sec:line}
In this section, we consider the line formation problem in the geometric amoebot model.
We assume that initially we have an arbitrary connected structure of contracted particles with a unique leader. The leader is used as the starting point for forming the line of particles and specifies the direction along which this line will grow. As the line grows, every particle touching the line that is already in a valid line position becomes part of the line. Any other particle adjacent to the line becomes the root of a tree of particles.
Every root aims at traveling around the line in a clockwise manner until it joins the line. As a root particle moves, the other particles in its tree follow in a worm-like fashion (i.e., via a series of handover operations)\footnote{For a simulation video of the Line Formation Algorithm please see http://sops.cs.upb.de .}.

Before we give a detailed description of the algorithm, we provide some preliminaries.
We distinguish for the state of a particle between \emph{idle}, \emph{follower}, \emph{root}, and \emph{retired} (or halted).
Initially, all particles are \emph{idle}, except for the leader particle, which is always in a {\em retired} state.
In addition to its state, each particle $p$ may maintain a constant number of {\em flags} in its shared memory.
For an expanded particle, we denote the node the particle last expanded into as the \emph{head} of the particle and call the other occupied node its
\emph{tail}:
In our
algorithm, we assume that every time a particle contracts, it contracts out of its tail.

The spanning forest algorithm (see Algorithm~\ref{alg:spanningForestAlgorithm}) is a basic building block we use for shape formation problems. This algorithm aims at organizing the particles as a spanning forest, where the particles that represent the roots of the trees determine the direction of movement, whom the remaining particles follow. Each particle $p$ continuously runs Algorithm~\ref{alg:spanningForestAlgorithm} until it retires. If particle $p$ is a follower, it stores a flag $p.parent$ in its shared memory corresponding to the edge adjacent to its parent $p'$ in the spanning forest (any particle $q$ can then easily check if $p$ is a child of $q$).
If $p$ is the leader particle, then it sets the flag $p.linedir$ in the shared memories corresponding to two of its edges in opposite directions (i.e., an edge $i$ and the edge $i+3$ (mod 6) in clockwise order), denoting that it would like to extend the line through the directions given by these edges.
\begin{algorithm}[htb]
    {\sc SpanningForest} $(p)$:

    \noindent Particle $p$ acts as follows, depending on its current state:

        \begin{tabularx}{\columnwidth}{lX}
        \textbf{idle}: &
                 If $p$ is connected to a retired particle, then $p$ becomes a {\em root} particle. Otherwise,
             if an adjacent particle $p'$ is a root or a follower,
        $p$ sets the flag $p.parent$ on the shared memory corresponding to the edge to $p'$ and becomes a {\em follower}. If none of the above applies, it remains idle.
        \\
        \textbf{follower}: &
        If $p$ is contracted and connected to a retired particle, then $p$ becomes a {\em root} particle.
         Otherwise, it considers the following three cases: $(i)$ if $p$ is contracted and $p$'s parent $p'$  (given by the flag $p.parent$) is expanded, then $p$ expands in a handover with $p'$, adjusting $p.parent$ to still point to  $p'$ after the handover;
        $(ii)$ if $p$ is expanded and has a contracted child particle $p'$,
        then $p$ executes a handover with $p'$; $(iii)$ if $p$ is expanded,  has no children,
        and $p$ has no idle neighbor, then $p$ contracts.
        \\
        \textbf{root}: &
        Particle $p$ may become {\em retired} following  {\sc CheckRetire} $(p)$. Otherwise, it considers the following three cases: $(i)$ if $p$ is contracted, it tries to expand in the direction given by {\sc LineDir}$(p)$; $(ii)$ If $p$ is expanded and has a child $p'$, then
            $p$ executes a handover contraction with $p'$;
         $(iii)$ if $p$ is expanded and has no children,
        and no idle neighbor,   then $p$ contracts.
        \\
        \textbf{retired}: &
            $p$ performs no further action.
        \\
    \end{tabularx}

		\vspace{.05in}
    {\sc CheckRetire} $(p)$: \\
		\vspace{-.15in}
   \begin{algorithmic}
    \If {$p$ is a contracted root}
    \If {$p$ has an adjacent edge $i$ to $p'$ with a flag $p'.linedir$, where $p'$ is retired}
    \State Let $i'$ be the edge opposite to $i$ in clockwise order
    \State $p$ sets the flag $p.linedir$ in the shared memory of edge $i'$ 
    \State $p$ becomes {\em retired}.
    \EndIf
    \EndIf
    \end{algorithmic}

     \vspace{.05in}
    {\sc LineDir($p$)}:\\
		\vspace{-.15in}
    \begin{algorithmic}
        \State Let $i$ be the label of an edge connected to a retired particle.
        \While{edge $i$ points to a retired particle}
        \State $i \; \gets \;$ label of next edge in clockwise direction
        \EndWhile
        \State \textbf{return} $i$
    \end{algorithmic}

    \caption{Line Formation Algorithm}
    \label{alg:spanningForestAlgorithm}
\end{algorithm}

We have the following theorem, where \emph{work} is defined as the number of expansions and contractions executed by all particles in the system:
\begin{theorem}
    Algorithm~\ref{alg:spanningForestAlgorithm} correctly decides the line formation problem in worst-case optimal $\mathcal{O}(n^2)$ work.
    \label{thm:line}
\end{theorem}

In order to prove Theorem~\ref{thm:line}, we first prove that the algorithm will eventually correctly converge to a line in $\Geqt$ in Sections~\ref{sec:spforest} and \ref{sec:lineconvergence}, and then show that the algorithm terminates within worst-case optimal $\mathcal{O}(n^2)$ work in Section~\ref{sec:linework}.

\subsection{Spanning Forest Formation}
\label{sec:spforest}
The first three lemmas demonstrate some properties that hold during the execution of the spanning forest algorithm  and will be used in Section~\ref{sec:lineconvergence} to analyze the complete algorithm, when we incorporate the check for retirement of particles according to the line formation problem,  and the propagation of the line direction.

The configuration of the system of particles at time $t$ consists, for every particle $p$, of the current state of $p$, including whether the particle is expanded or contracted, any flags in $p$'s shared memory, the
node(s) $p$ occupies in $G$ (given by the relative position of $p$ according to the other particles) at time $t$, as well as the labeling of the bonds of $p$.
Following the standard asynchronous model, the system of particles progresses by performing atomic actions, each of which affects the configuration
 of one or two particles. We say that 	 followers and roots
particles are \emph{active}.
As specified in Algorithm~\ref{alg:spanningForestAlgorithm}, only followers can set the flag $p.parent$.
\begin{lemma}
    \label{lem:successor}
    For every follower $p$, the node indicated by the flag $p.parent$ is occupied by a non-idle particle.
\end{lemma}
\begin{proof}
    Consider a follower $p$ in any configuration during the execution of Algorithm~\ref{alg:spanningForestAlgorithm}.
    Note that $p$ can only become a  follower from an idle state,
    and once it leaves the follower state it will not switch back to that state again.
    Consider the first configuration $c_1$ in which $p$ is a follower.
    In the configuration $c_0$ immediately before $c_1$, $p$ must be idle
    and it becomes a follower because of an active particle $p'$ occupying the position indicated by $p.parent$ in $c_0$.
    The particle $p'$ is still adjacent to the edge flagged by $p.parent$ in $c_1$.
    Now assume that $p.parent$ points to an active particle $p'$ in a configuration $c_i$,
    and that $p$ is still a follower in the next configuration $c_{i+1}$ that results from executing an action $a$.
    If $a$ affects $p$ and $p'$, the action must be a handover in which $p$ updates its flag $p.parent$
    such that $p.parent$ may be moved to the edge that now connects $p$ to $p'$ in $c_{i+1}$.
    If $a$ affects $p$ but not $p'$,
    it must be a contraction in which $p.parent$ does not change and still points to $p'$.
    If $a$ affects $p'$ but not $p$, there are multiple possibilities.
    The particle $p'$ might switch from follower to root state,
     or from root to retired state,
        or it might expand,
    none of which violates the lemma.
    Furthermore, $p'$ might contract.
    If $p.parent$ points to the head of $p'$, $p'$ is still adjacent to the edge flagged by $p.parent$ in $c_{i+1}$.
    Otherwise, $p$ is a child adjacent to the tail of $p'$ in $c_i$ and therefore the contraction must be part of a handover.
    As $p$ is not involved in the action, the handover must be between $p'$ and a third active particle $p''$.
    After the handover, $p''$ will occupy the position originally occupied by the tail of $p'$ and hence $p.parent$ points to $p''$. 
    Finally, if $a$ affects neither $p$ nor $p'$, $p.parent$ will still point to $p'$ in $c_{i+1}$.
    \qed
\end{proof}

Based on Lemma~\ref{lem:successor},
$A(c)$ contains the same nodes as the nodes occupied in $\Geqt$ by the set of active particles in $c$.
For every expanded particle $p$ in $c$, $A(c)$ contains a directed edge from the tail to the head of $p$, and for every follower $p'$ in $c$, $A(c)$ contains a directed edge from the head of $p'$ to $p'.parent$, if $p'.parent$ is occupied by an active particle.

\begin{lemma}
    \label{lem:forest}
    The graph $A(c)$ is a forest, 
    and every connected component of idle particles 
		 is connected to a non-idle  particle.
\end{lemma}
\begin{proof}
		 Since in an initial configuration $c_0$ all particles are idle, except for the leader particle which is retired,  and the particle system is connected, the lemma holds trivially for $c_0$.
    Now assume that the lemma holds up to a cerain configuration $c_i$ and consider a connected component of idle particles.
    If an idle particle $p$ in the component is activated, it may stay idle or become an active particle.
    The former case does not affect the configuration. So consider  the latter case. If $p$ has changed into a follower state, it joins an existing tree, and in case $p$ has become a root, it forms a new tree in $A(c_{i+1})$.
    In either case, $A(c_{i+1})$ is a forest and the connected component of idle particles
    that $p$ belongs to in $c_i$ is either non-existent or connected to $p$ in $c_{i+1}$.
    If a follower or a root particle $p$ that is connected to the idle component is activated,  it cannot contract
		 unless by a handover with another active particle, which implies the nodes occupied by $p$ will remain occupied by the active
		particles.  While such a handover can change the parent relation among the nodes, it cannot violate the lemma.
		 If an active particle $p$ is activated and has no child $p'$ such that $p'.parent$ is the tail of $p$ and $p$ is not connected to the idle component, it contracts and its contraction does not disconnect
		any particle in its tree in $A(c_{i+1})$.
		Moreover, an expansion of an active particle or changing its state to a root, if it is a follower, or to a retired particle, if it is a root particle, does not violate the lemma too.
		 Finally, if  a retired particle 
			 is activated, it does not move.
		 Therefore, the lemma holds at configuration $c_{i+1}$.

    \qed
\end{proof}

The following lemma shows that the spanning forest always makes progress, by showing that as long as the roots keep moving, the remaining particles will eventually follow.

\begin{lemma}
    \label{lem:contract}
    An expanded particle eventually contracts.
\end{lemma}
\begin{proof}
    Consider an expanded particle $p$ in a configuration $c$.
    Note that $p$ must be active.
    If there is an enabled action that includes the contraction of $p$,
    that action will remain enabled until $p$ eventually contracts  when it is selected within the current or next round.
      So assume that there is no enabled action that includes the contraction of $p$.
    According to Lemma~\ref{lem:forest} and the transition rule from idle to some active state,
    at some point in time there will be no idle particles in the system.
    If the contraction of $p$ becomes part of an enabled action before this happens, $p$ will eventually contract.
    So assume that all particles are non-idle but still $p$ cannot contract.
    If $p$ has no children, the isolated contraction of $p$ is an enabled action which contradicts our assumption.
    Therefore, $p$ must have children. 
    Furthermore, at least one child $p'$ of $p$ must have a $p'.parent$ flag on the edge connecting to the tail of $p$ and all children of $p$ must be expanded,
    as otherwise $p$ could again contract as part of a handover.
    If $p'$ would contract, a handover between $p'$ and $p$ would become an enabled action and $p$ would eventually contract. Hence $p'$ also cannot contract. Applying this argument recursively, we identify
to $p'$
    a set of expanded particles forming a branch of a tree in $A(c)$, until we reach a leaf $q$ of $A(c)$ (by Lemma~\ref{lem:forest}). Obviously $q$ will contract next time it is activated.
    Therefore, we found a sequence of expanded particles that starts with $p'$
    and ends with a particle that eventually contracts, implying that
it in the sequence to contract and so on.
    the contraction of $p$ will become part of an enabled action and therefore $p$ will eventually contract.
    \qed
\end{proof}

\subsection{Line Formation}
\label{sec:lineconvergence}
Now, we can show that the algorithmic primitives as developed in the Section~\ref{sec:lineFormation} solve the line formation problem.

\begin{theorem}
    \label{thm:solve}
    Algorithm~\ref{alg:spanningForestAlgorithm} correctly decides the line formation problem.
\end{theorem}
\begin{proof}
    We need to show that the algorithm terminates and that when it does, the formed shape is a straight line in $\Geqt$.
        As a result of Lemma~\ref{lem:forest} and due to  the transition rules from idle to active states, every idle particle $p$ eventually switches out of idle state.
    According to the algorithm proposed for the line formation problem, if $p$ is adjacent to a retired particle, it becomes a root and moves in clockwise order  around the current line structure,
		 until it eventually reaches one of the valid positions that can extend the line and becomes retired (halted). By contradiction, assume $p$ never becomes retired. Since the number of particles is bounded (and therefore the size of the current line structure  is bounded), $p$ cannot move around the line structure indefinitely (since in this case $p$ must occupy a valid extension position and would have become retired). Thus there must be an infinite number of configurations $c_i$ where $p$ has a particle blocking its desired path around the current line structure.
		Let $p'$ be the last particle $p$ sees as its clockwise neighbor over the line structure. Since $p'$ is touching a retired particle, $p'$ will become a root particle within at most two rounds, and will stay connected to the line structure and always attempt to move in a clockwise manner, $p'$ is well-defined.
     Applying the same argument we had for $p$ inductively to $p'$ results an infinite sequence of roots adjacent to the line structure that never touch a valid spot pointed by one of  $p.linedir$ flags of an already retired particle belong to the line. This is a contradiction, since the current line structure (the current number of retired particles) is bounded.
     Therefore, every root eventually changes into a retired state.

According to spanning forest construction, as long as the roots keep moving the followers eventually follow.
 From Algorithm~\ref{alg:spanningForestAlgorithm}, every follower in the
 neighborhood of a retired particle becomes a root. For every root $q$ with at least one follower child, let $c$ be the first
configuration when $q$ becomes retired.
If $q$ still has any child
 in $c$ then all of its children $p$
become roots. Applying this argument recursively together with the already proven fact that every root eventually becomes retired we will reach to a configuration such that there exists no root with a follower child
 which proves that eventually every follower becomes a root.

 Putting it all together, eventually all particles become retired and the algorithm terminates.
 Note that it also follows from the argument above that the set of retired particles at the end of the algorithm forms a connected structure (since the particles start from a connected configuration and never get disconnected through the process). The connected structure must form a line, since a root particle may only become retired once it is contracted and occupies a valid spot extending the current line.
\qed
\end{proof}

\subsection{Performed Work}
\label{sec:linework}
Finally, we evaluate the performance of our algorithm in terms of the number of movements (expansions and contractions) of the particles, i.e., the total {\em work} performed by the algorithm.

\begin{lemma}
    \label{lem:worstCaseLine}
    The worst-case work required by any algorithm to solve the line formation problem is $\Omega(n^2)$.
\end{lemma}
\begin{proof}
Assume the initial configuration of the set of particles forms a connected structure of diameter at most $2\sqrt{n}+2$ (e.g., if it forms a hexagonal or square shape in $\Geqt$). Since the line has diameter equal to $n-1$, there must exist some particle that will need to traverse a distance of at least $\frac 1 2 n-2\sqrt{n}-3$,  a second particle that will traverse a distance of $n-2\sqrt{n}-4$, etc., irrespective of the algorithm used. The number of particle movements incurred will be at least $\sum_{i=1}^{\frac 1 2 n-2\sqrt{n}-1} (\frac 1 2 n-2\sqrt{n}-i-2)=\Theta(n^2)$.
 \qed
\end{proof}

In the following, we will show a matching upper bound:
\begin{theorem}
Algorithm~\ref{alg:spanningForestAlgorithm} terminates in $\mathcal{O}(n^2)$ work.
\label{thm:work}
\end{theorem}

\begin{proof}
To prove the upper bound, we simply show that every particle executes $\mathcal{O}(n)$ movements. The theorem then follows. 
    Consider a particle $p$.
    While $p$ is in an idle or a retired state, it does not move.
    Let $c$ be the first configuration when $p$ becomes a follower.
     Consider the directed path in $A(c)$ from the head of $p$ to its root $p'$. There always is a such a path since every follower belongs to a tree in $A(c)$ by
    Lemma~\ref{lem:forest}.
    Let $P= (a_0, a_1, \ldots, a_m)$ be that path in $A(c)$ where $a_0$ is the head of $p$ and $a_m$ is a child of  $p'$. According to Algorithm~\ref{alg:spanningForestAlgorithm}, $p$ will follow $P$ by sequentially expanding into the nodes $a_0, a_1, \ldots, a_m$. The length of this path is bounded by $2n$ and, therefore, the number of movements $p$ executes while being a follower is $\mathcal{O}(n)$.
     Once $p$ becomes a root, it only performs expansions and contractions around the currently constructed line structure, 
     until it reaches one of the valid positions on the line. Since the total number of retired particles is at most $n$, this leads to an additional $\mathcal{O}(n)$ movements by $p$.
     Therefore, the number of movements a particle $p$ executes is $\mathcal{O}(n)$, which concludes the theorem.
     \qed
\end{proof}

\section{Self-stabilizing Leader Election and Shape Formation}
\label{sec:ss}
Consider the variant of the geometric amoebot model in which faults can occur
that arbitrarily corrupt the local memory of a particle. Recall that for an
algorithm to solve the leader election problem in a self-stabilizing manner,
it has to satisfy the following requirements: First, from any initial system
state (in which the particle structure is connected) the particle system
eventually reaches a final system state while preserving connectivity, i.e.,
eventually a unique leader will be established. Second, once a final system
state is reached, the system has to remain in that state as long as no faults
occur. Analogous requirements have to be satisfied for self-stabilizing shape
formation.

Our leader election algorithm can be extended to a self-stabilizing leader
election algorithm with $\mathcal{O}(\log^* n)$ memory using the results
of~\cite{DBLP:conf/podc/AwerbuchO94,DBLP:conf/focs/ItkisL94} (i.e., we use
their self-stabilizing reset algorithm on every cycle in order to recover from
failure states). However, it is not possible to design a self-stabilizing
algorithm for the line formation. The reason for this is that even a much
simpler problem called \emph{movement problem} cannot be solved in a
self-stabilizing manner. It is easy to see that if the movement problem cannot
be solved in a self-stabilizing manner, then also the line formation problem
cannot be solved in a self-stabilizing manner.

In the movement problem we are given an initial distribution $A$ of particles
that can be in a contracted as well as expanded state, and the goal is to
change the set of nodes occupied by the particles without causing
disconnectivity. For the ring of expanded particles it holds that for any
protocol $P$ there is an initial state so that $P$ does not solve the movement
problem. To show this we consider two cases: suppose that there is any state
$s$ for some particle in the ring that would cause that particle to contract.
In this case set two particles on opposite sides of the ring to that state,
and the ring will break due to their contractions. Otherwise, $P$ would not
move any particle of the ring, so also in this case it would not solve the
movement problem in a self-stabilizing manner.

\section{Conclusion}
The algorithms presented for the geometric amoebot model can be extended for the case that $G$ is a different regular grid graph embedded in the two-dimensional Euclidean plane.
As future work, we would like to identify the minimum set of key geometric
properties that $G$ must have in order for the proposed algorithms to work.

\section*{Acknowledgment} We would like to thank John Reif for the helpful discussions that led to the realization of our algorithm with high probability guarantees.
\bibliographystyle{plain}
\bibliography{literature}

\end{document}